\documentclass{eptcs} 
\pdfoutput=1

\providecommand{\event}{M. Boreale, S. Kremer (Eds.):\\
 Security Issues in Concurrency (SecCo'09)}
\providecommand{\volume}{7}
\providecommand{\anno}{2009}
\providecommand{\firstpage}{46}
\providecommand{\eid}{4}

\usepackage{breakurl}

\newif\ifeptcs
\eptcstrue		

\newif\iffsttcs
\fsttcsfalse		

\newif\ifieee
\ieeefalse		

\newif\ifllncs
\llncsfalse		

\newif\ifamsart
\amsartfalse             

\ifllncs\usepackage{smallsubsub}\fi
\ifieee
\usepackage{latex8}
\usepackage{times}
\fi

\usepackage{latexsym}
\usepackage{epic}
\usepackage{epsfig}
\usepackage{amsmath}
\usepackage{amssymb}
\usepackage{amscd}
\usepackage{url}
\usepackage{color}
\usepackage{mathpartir}

\usepackage{graphics}

\input{xy}
 \xyoption{all} 

 \newcommand{\myparagraph}[1]{\medskip\noindent\textbf{#1} }

 \newcommand{\kind}[1]{\ensuremath{\mathsf{#1}}}
 \newcommand{\tagname}[1]{\ensuremath{\mathit{#1}\;}}

 \newcommand{\Algebra}{\kind{A}}

 \renewcommand{\SS}{\kind{S}} 
 \newcommand{\vr}[1]{\mathit{{#1}}}
 \newcommand{\opr}[1]{\operatorname{\kind{#1}}}

 \newcommand{\EM}{\opr{EM}}
 \newcommand{\EK}{\opr{EK}}
 \newcommand{\EOO}{\opr{EOO}}
 \newcommand{\EOR}{\opr{EOR}}
 \newcommand{\AR}{\opr{AR}}
 \newcommand{\AT}{\opr{AT}}
 \newcommand{\RR}{\opr{RR}}

 \newcommand{\D}{\opr{D}}
 \newcommand{\CF}{\opr{CF}}

 \ifamsart{}\else{
 
 }\fi

 \newcommand{\bnd}{\mathcal{B}}

 \newcommand{\qdot}{\; . \;} 
 \newcommand{\cons}{\,{\hat{\ }}\,}
 \newcommand{\seq}[1]{\langle #1 \rangle}

 \newcommand{\strandnode}[2]{ {#1}\downarrow{#2} }

 \newcommand{\Hear}[2]{\kind{Lsn{#1}}[#2]}

 \newcommand{\enc}[2]{\{\!|#1|\!\}_{#2}}
 \newcommand{\encR}[3]{\{\!|#1|\!\}_{#2}^{#3}}
 \newcommand{\tagged}[2]{[\![\,#1\,]\!]_{#2}}

 \newcommand{\subterm}{\sqsubseteq}
 \newcommand{\term}{\kind{msg\/}}

 \newcommand{\applyrep}[2]{\ensuremath{#1 \cdot #2}}
 \newcommand{\applysub}[2]{\ensuremath{#1 \cdot #2}}

 \newcommand{\multiplicity}[2]{\ensuremath{|#1|_{#2}}}

 \newcommand{\pleads}[2]{\ensuremath{{#2} (#1)}}
 
 \newcommand{\supfrown}{\ensuremath{{}^{\frown}}}

 \ifllncs{
 \newtheorem{prop}{Proposition}{\bfseries}{\itshape}
 \newtheorem{mylemma}[prop]{Lemma}{\bfseries}{\itshape}
 \newtheorem{thm}{Theorem}{\bfseries}{\itshape}
 \newtheorem{eg}{Example}{\bfseries}{\upshape}
 \newtheorem{cor}[prop]{Corollary}{\bfseries}{\itshape}
 \institute{The MITRE Corporation
 }}\else
 {
 \iffsttcs{
 \newtheorem{prop}{Proposition}[section]
 \newtheorem{cor}[prop]{Corollary}
 \newtheorem{thm}[prop]{Theorem}
 }\else
 \newtheorem{prop}{Proposition}[section]

 \newtheorem{definition}[prop]{Definition}
 
 \newtheorem{cor}[prop]{Corollary}
 \newtheorem{lemma}[prop]{Lemma}
 \newtheorem{thm}[prop]{Theorem}
 
 \newtheorem{Assum}{Assumption}

 \ifamsart{}\else

 \def\squareforqed{\hbox{\rlap{$\sqcap$}$\sqcup$}}
 \def\qed{\ifmmode\squareforqed\else{\unskip\nobreak\hfil
 \penalty50\hskip1em\null\nobreak\hfil\squareforqed
 \parfillskip=0pt\finalhyphendemerits=0\endgraf}\fi}

 \newenvironment{proof}{\emph{Proof. }}{{\hspace*{\fill}\qed}}
 \fi 
 \fi}
 \fi

 \ifeptcs
 \def\squareforqed{\hbox{\rlap{$\sqcap$}$\sqcup$}}
 \def\qed{\ifmmode\squareforqed\else{\unskip\nobreak\hfil
 \penalty50\hskip1em\null\nobreak\hfil\squareforqed
 \parfillskip=0pt\finalhyphendemerits=0\endgraf}\fi}

 \newenvironment{proof}{\emph{Proof. }}{{\hspace*{\fill}\qed}}
 \fi

 \newcommand{\h}{\ensuremath{\mathsf{h}}}

 \iffsttcs{}\else
 \title{Fair Exchange in Strand Spaces\thanks{Funded by MITRE-Sponsored
     Research.  Email:  \texttt{guttman@mitre.org}.}}

 \author{Joshua D.~Guttman\\ The MITRE Corporation and \\ Worcester
   Polytechnic Institute }
 \fi

\def\titlerunning{Fair Exchange in Strand Spaces}
\def\authorrunning{Joshua D.~Guttman}
 
\begin{document}

 \maketitle




 \nocite{ChadhaEtAl03}

 \begin{abstract}

   Many cryptographic protocols are intended to \emph{coordinate state
     changes} among principals.  Exchange protocols coordinate delivery
   of new values to the participants, e.g.~additions to the set of
   values they possess.  An exchange protocol is \emph{fair} if it
   ensures that delivery of new values is balanced:  If one participant
   obtains a new possession via the protocol, then all other
   participants will, too.  Fair exchange requires \emph{progress}
   assumptions, unlike some other protocol properties.

   The strand space model is a framework for design and verification of
   cryptographic protocols.  A \emph{strand} is a local behavior of a
   single principal in a single session of a protocol.  A \emph{bundle}
   is a partially ordered global execution built from protocol strands
   and adversary activities.  

   The strand space model needs two additions for fair exchange
   protocols.  First, we regard the state as a multiset of facts, and
   we allow strands to cause changes in this state via multiset
   rewriting.  Second, progress assumptions stipulate that some
   channels are resilient---and guaranteed to deliver messages---and
   some principals are assumed not to stop at certain critical steps.

   This method leads to proofs of correctness that cleanly separate
   protocol properties, such as authentication and confidentiality,
   from invariants governing state evolution.  G. Wang's recent fair
   exchange protocol illustrates the approach.
 \end{abstract}


\section{Introduction} 
\label{sec:intro} 

Many cryptographic protocols are meant to \emph{coordinate state
  changes} between principals in distributed systems.  For instance,
electronic commerce protocols aim to coordinate state changes among a
customer, a merchant, and one or more financial institutions.  The
financial institutions should record credits and debits against the
accounts of the customer and the merchant, and these state changes
should be correlated with state changes at the merchant and the
customer.  The merchant's state changes should include issuing a
shipping order to its warehouse.  The customer records a copy of the
shipping order, and a receipt for the funds from its financial
institution.  The job of the designer of an application-level protocol
like this is to ensure that these changes occur in a coordinated,
transaction-like way.

State changes should occur only when the participants have taken
certain actions, e.g.~the customer must have authorized any funds
transfer that occurs.  Moreover, they should occur only when the
participants have certain joint knowledge, e.g.~that they all agree on
the identities of the participants in the transaction, and the amount
of money involved.  These are \emph{authentication} goals in the
parlance of protocol analysis.  There may also be
\emph{confidentiality} goals that limit joint knowledge.  In our
example, the customer and merchant should agree on the goods being
purchased, which should not be disclosed to the bank, while the
customer and bank should agree on the account number or card number,
which should not be disclosed to the merchant.

\myparagraph{Goal of this paper.}  In this paper, we develop a model
of the interaction of protocol execution with state and state change.
We use our model to provide a proof of a clever fair exchange protocol
due to Guilin Wang~\cite{Wang2006}, modulo a slight correction.  

We believe that the strength of the model is evident in the proof's
clean composition of protocol-specific reasoning with state-specific
reasoning.  In particular, our proof modularizes what it needs to know
about protocol behavior into the four authentication properties given
in Section~\ref{sec:example}, Lemmas~\ref{lemma:auth:init:resp}--\ref{lemma:auth:ttp}.  If any protocol achieves these
authentication goals and its roles obey simple conditions on the
ordering of events, then other details do not matter: it will succeed
as a fair exchange protocol.

A two-party fair exchange protocol is a mechanism to deposit a pair of
values atomically into the states of a pair of principals.  Certified
delivery protocols are a typical kind of fair exchange protocol.  A
certified delivery protocol aims to allow $A$, the sender of a
message, to obtain a digitally signed receipt if the message is
delivered to $B$.  $B$ should obtain the message together with signed
evidence that it came from $A$.  If a session fails, then neither
principal should obtain these values.  If it succeeds, then both
should obtain them.  The protocol goal is to cause state evolution of
these participants to be \emph{balanced}.

The ``fair'' in ``fair exchange'' refers to the balanced evolution of
the state.  ``Fair'' does not have the same sense as in some other
uses in computer science, where an infinitely long execution is
\emph{fair} if any event actually occurs, assuming that it is enabled
in an infinite subsequence of the states in that execution.  In some
frameworks, fairness in this latter sense helps to clarify the
workings of fair exchange
protocols~\cite{CederquistEtAl07,TorabiDashti07}.  However, we show
here how fair exchange protocols can also be understood independent of
this notion of fairness.  When we formalize Wang's
protocol~\cite{Wang2006}, we use an extension of the strand space
model~\cite{GuttmanThayer02} in which there are no infinite executions
or fairness assumptions.

As has been long known~\cite{EvenYacobi80,Rabin81}, a deterministic
fair exchange protocol must rely on a trusted third party $T$.  Recent
protocols generally follow~\cite{AsokanEtAl00} in using the trusted
third party optimistically, i.e.~$T$ is never contacted in the
extremely common case that a session terminates normally between the
two participants.  $T$ is contacted only when one participant does not
receive an expected message.

Each principal $A,B,T$ has a state.  $T$ uses its state to record the
sessions in which one participant has contacted it.  For each such
session, $T$ remembers the outcome---whether $T$ aborted the session
or completed it successfully---so that it can deliver the same outcome
to the other participant.  The states of $A,B$ simply records the
ultimate result of each session in which it participates.  The
protocol guides the state's evolution to ensure balanced changes.

\myparagraph{Strand space extensions.}  Two additions to strand
spaces are needed to view protocols as solving to coordinated state
change problems.  A \emph{strand} is a sequence of actions executed by
a single principal in a single local session of a protocol.

We enrich strands to allow them to \emph{synchronize} with the
projection of the joint state that is local to the principal $P$
executing the strand.  We previously defined the actions on a strand
to be either (1) message transmissions or (2) message receptions.  We
now extend the definition to allow the actions also to be \emph{(3)
  state synchronization events}.  $P$'s state at a particular time may
permit some state synchronization events and prohibit others, so that
$P$'s strands are blocked from the latter behaviors.  Thus, the state
constrains protocol behavior.  Updates to $P$'s state may record
actions on $P$'s strands.

We represent states by multisets of facts, and state change by
multiset rewriting~\cite{MitchellEtAl99,DurginEtAl04}, although with
several differences from Mitchell, Scedrov et al.  First, they use
multiset rewriting to model protocol and communication behavior, as
well as the states of the principals.  We instead use strands for the
protocol and communication behavior.  Our multiset rewriting
represents only changes to a single principal's local state.  Hence,
second, in our rules we do not need existentials, which they used to
model selection of fresh values.  Third, we tend to use ``big'' states
that may have a high cardinality of facts.  However, the big states
are generally sparse, and extremely easy to implement with small data
structures.

We also incorporate \emph{guaranteed progress} assumptions into strand
spaces.  Protocols that establish balance properties need guaranteed
progress.  Since principals communicate by messages, one of
them---call it $A$---must be ready to make its state change first.
Some principal (either $A$ or some third party) must send a message to
$B$ to enable it to make its state change.  If this message never
reaches $B$, $B$ cannot execute its state change.  Hence, in the
absence of a mechanism to ensure progress, $A$ has a strategy---by
preventing future message deliveries---to prevent the joint state from
returning to balance.  

These two augmentations---state synchronization events and a way to
stipulate progress---fit together to form a strand space theory usable
for reasoning about coordinated state change.

\myparagraph{Structure of this paper.}  Section~\ref{sec:example}
describes Wang's protocol.  Two lemmas
(Lemmas~\ref{lemma:auth:init:resp} and~\ref{lemma:auth:ttp}) summarize
the authentication properties that we will rely on.  Any protocol
whose message flow satisfies these two lemmas, and which synchronizes
with state history at the same points, will meet our needs.

Section~\ref{sec:state} introduces our multiset rewriting framework,
proving a locality property.  This property says that state
synchronization events of two different principals are always
concurrent in the sense that they commute.  Hence, coordination
between different principals can only occur by protocol messages, not
directly by state changes.  We also formalize the state facts and
rules for Wang's protocol, inferring central facts about computations
using these rules.  These (very easily verified) facts are summarized
in Lemma~\ref{lemma:GW:computation}.  Any system of rules that
satisfies Lemma~\ref{lemma:GW:computation} will meet our needs.

Section~\ref{sec:progress} gives definitions for guaranteed progress,
applying them to Wang's protocol.  Lemma~\ref{thm:progress}, the key
conclusion of Section~\ref{sec:progress}, says that any compliant
principals executing a session with a session number $L$ can always
proceed to the end of a local run, assuming only that the trusted
third party is ``ready'' to handle sessions labeled $L$.

In Section~\ref{sec:correctness} we put the pieces together to show
that it achieves its balanced state evolution goal.  In particular,
the balance property depends only on Lemmas~\ref{lemma:auth:init:resp}
and~\ref{lemma:auth:ttp} about the protocol structure,
Lemma~\ref{lemma:GW:computation} about the state history mechanism,
and lemma~\ref{thm:progress} about progress.  In this way, the
verification is well-factored into three sharply distinguished
conceptual components.


\section{The Gist of Wang's Protocol}
\label{sec:example}

Wang's fair exchange protocol~\cite{Wang2006} is appealing because it
is short---only three messages in the main exchange
(Fig.~\ref{fig:wang:exch})---and uses only ``generic'' cryptography.
By generic cryptography, Wang means standard digital signatures, and
probabilistic asymmetric encryption such that the random parameter may
be recovered when decryption occurs.  RSA-OAEP is such a scheme.  In
many situations, these advantages will probably outweigh one
additional step in the dispute resolution (see below in this section,
p.~\pageref{point:dispute}).

We write $\enc{t}{k}$ for $t$ encrypted with the key $k$, and
$\encR{t}{k}{r}$ for $t$ encrypted with the key $k$ using recoverable
random value $r$.  We write $\h(t)$ for a cryptographic hash of $t$,
and $\tagged{t}{k}$ for a digital signature on $t$ which may be
verified using key $k$.  By this, we mean $t$ together with a
cryptographic value prepared from $\h(t)$ using $k^{-1}$, the private
signature key corresponding to $k$.  When we use a principal name
$A,B,T$ in place of $k$, we mean that a public key associated with
that principal is used for encryption, as in $\encR{t}{T}{r}$, or for
signature verification, as in $\tagged{t}{A}$.  Message ingredients
such as $\opr{keytag}, \opr{ab\_rq}, \opr{ab\_cf}$, etc., are
distinctive bit-patterns used to tag data, indicate requests or
confirmations, etc.  Our notation differs somewhat from Wang's; for
instance, his $L$ is our $\h(L)$.

\myparagraph{Main exchange.} In the first message
(Fig.~\ref{fig:wang:exch}), $A$ sends the payload $M$ to $B$ encrypted
with a key $K$, as well as $K$ encrypted with the public encryption
key of the trusted third party $T$.
\begin{figure}\hrule \vspace{2mm}
  $$
  \begin{array}{r@{\rightarrow}r@{\colon\quad}l}
    A & B & L\cons\EM\cons\EK\cons\EOO  \\[2mm]
    B & A & \EOR \\[2mm]
    A & B & K\cons R
  \end{array}
  $$
  $$
  \begin{array}[c]{r@{\qquad}c@{\qquad}c}
    \mbox{where:} & L=A\cons B\cons T\cons{\h(\EM)}\cons\h(K) & 
    \EM=\enc{M}K \\
    \EK=\encR{\opr{keytag}\cons \h(L)\cons K}{T}{R} &
    \EOO=\tagged{\opr{eootag}\cons \h(L)\cons \EK}A  &
    \EOR=\tagged{\opr{eortag}\cons \h(L)\cons\EK}B
  \end{array}
$$
  \caption{Wang's protocol:  A Successful Run}
  \label{fig:wang:exch}
\vspace{2mm}\hrule 
\end{figure}
$A$ also sends a digitally signed unit $\EOO$ asserting that the
payload (etc.)  originate with $A$.  The value $L$ serves to identify
this session uniquely.
%
%
In the second message, $B$ countersigns $\h(L),\EK$.  In the third
message, $A$ discloses $K$ and the random value $R$ used originally to
encrypt $K$ for $T$.  $B$ uses this information to obtain $M$, and
also to reconstruct $\EK$, and thus to validate that the hashes inside
$\EOO$ are correctly constructed.  At the end of a successful
exchange, each party deposits the resulting values as a record in its
state repository.

\myparagraph{Abort and recovery subprotocols.} What can go wrong?  If
the signature keys are uncompromised and the random values $K,R$ are
freshly chosen, only two things can fail.  Either $A$ fails to receive
$B$'s countersigned evidence $\EOR$; or else $A$ receives it, but $B$
fails to receive a correct $K,R$.
\begin{enumerate}
  \item If $A$ fails to receive $\EOR$, then $A$ sends the session
  identifier $L$ and a signed abort request $\AR$ to $T$.  $T$ may
  confirm, and certify the session is aborted, sending a countersigned
  $\tagged{\AR}T$.
  \item If $B$ sends $\EOR$ but does not receive $K,R$, then $B$ asks
  $T$ to ``recover'' the session.  To do so, $B$ sends
  $L\cons\EK\cons\EOO\cons\EOR$ to $T$, inside a signed unit $\RR$
  indicating that this is a recovery request.

  $T$ can now decrypt the encrypted key $\EK=\encR{\opr{keytag}\cons
    \h(L)\cons K}{T}{R}$, returning $K\cons R$.  If $T$'s attempt to
  decrypt fails, or yields a values incompatible with the session
  information, then no harm is done:  $A$ will never be able to
  convince a judge that a valid transaction occurred.  Wang's protocol
  returns an error message that we do not show
  here~\cite[Fig.~3]{Wang2006}.
\end{enumerate}
What should happen if $A$ makes an abort request and $B$ also makes a
recovery request, perhaps because $\EOR$ was sent but lost in
transmission?  $T$ services whichever request is received first.  When
the other party's request is received, $T$ reports the result of that
first action.  
\begin{figure}[b]
  \centering
  $$\xymatrix@C=8mm@R=5mm{
    & \bullet\ar[r]^{\D}\ar@{=>}[d]\ar@{=>}[dr] & & A & &
    \quad\ar[r]^{\D} \ar@{.}[dddd]
    &  \bullet\ar@{=>}[d] &    & B \\
    \null\ar[r]^{\EOR} &\bullet\ar@{=>}[d]\ar@{=>}[r] &
    \bullet\ar[r]^{\AR}\ar@{=>}[d]\ar@{=>}[drr] & \null & & \quad
    & \bullet\ar[l]_{\EOR}\ar@{=>}[d]\ar@{=>}[dr] &       &       \\
    \null &\bullet\ar[l]_{K\cons R}\ar@{=>}[d] &\bullet\ar@{=>}[d] &
    \quad\null\quad\ar[l]_{\AT} \ar[r]_{\EOR} & \;\bullet\;\ar@{=>}[d]
    & \quad\ar[r]^{K\cons R} &
    \bullet\ar@{=>}[d] & \bullet\ar[r]^{\RR}\ar@{=>}[d]\ar@{=>}[dr] &       \\
    &\ar@{}[l]|{\mathsf{depEOR()\quad}}\circ &
    \circ\ar@{}[r]|{\mathsf{depAT()}} &&
    \circ\ar@{}[l]|{\mathsf{depEOR()}} & \ar@{}[r]|{\mathsf{depEOO()}}&
    \quad\circ\quad\ar@{ ->}[r]^{K\cons R} & \bullet\ar@{=>}[d] &
    \bullet\ar@{=>}[d] & \ar[l]_{\AT}
    \\
    &&&& &&&\circ\ar@{}[l]|{\mathsf{depEOO()}} &
    \circ\ar@{}[r]|{\mathsf{\quad depAT()}} & }
  $$
  \caption{Initiator (A) and Responder (B) Behavior}
  \label{fig:wang:init:resp}
\end{figure}
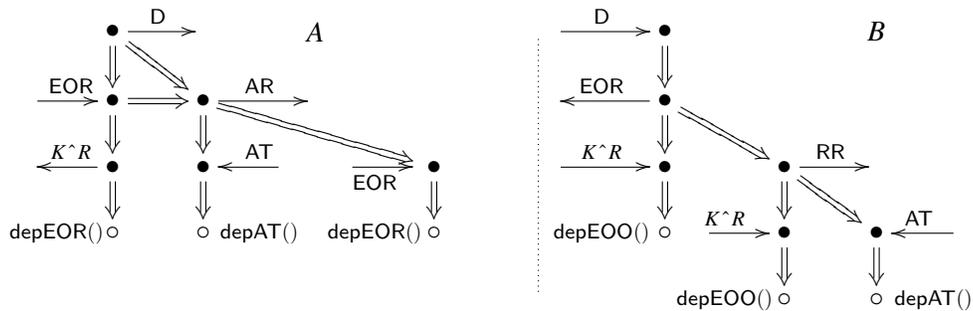
The local behaviors (strands) for $A,B$ in this protocol are shown in
Fig.~\ref{fig:wang:init:resp}.  The local sessions (strands) are the
paths from a root to a terminal node; there are four paths for $A$ and
three paths for $B$.  The solid nodes indicate messages to be sent or
received, while the hollow nodes $\circ$ indicate events in which the
participants deposit results into their state repositories.  This
figure is not precise about the forms of the messages, the parameters
available to each participant at each point in its run, or the
parameters to the state synchronization events.  For instance, $B$
does not know whether a claimed $\EM$ is really of the form $\enc{M}K$
when first receiving it, nor what $M,K$ would produce the message
received.  However, the fairness of the protocol is largely
independent of these details.

$A$'s abort request $\AR$ elicits an abort confirmation
$\tagged{\AR}T$ if it reaches $T$ first, but it elicits a recovery
token $L\cons\EOR$ if $B$'s recovery request was received first.
Likewise, $B$'s recovery request $\RR$ elicits $K\cons R$ if it is
received first, but it elicits the abort confirmation $\tagged{\AR}T$
if $A$'s abort request was received first.  $T$ must synchronize with
its state to ensure that these different requests are serviced in
compatible ways, depending on whichever arrived first.  This
compatibility of responses ensures that $A,B$ will execute balanced
state changes.

These behaviors of the trusted third party $T$, together with an
additional behavior concerned with dispute resolution, are summarized
in Fig.~\ref{fig:wang:ttp}.
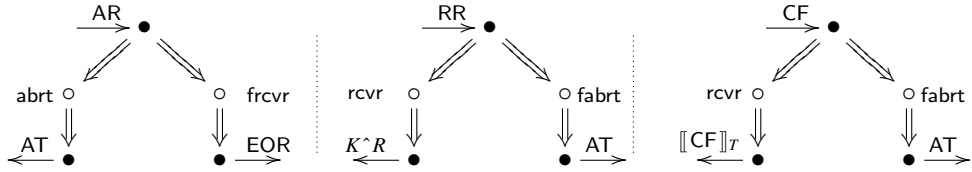
\begin{figure}
  \centering
  $$\xymatrix@C=6mm@R=5mm{
    & \null\ar[r]^{\AR} & \bullet\ar@{=>}[dl]\ar@{=>}[dr] &  &  \ar@{..}[dd]
    & \null\ar[r]^{\RR} & \bullet\ar@{=>}[dl]\ar@{=>}[dr] &  &  \ar@{..}[dd]\\
    & \ar@{}[l]|{\mathsf{abrt}}\circ\ar@{=>}[d] & 
    & \circ\ar@{}[r]|{\mathsf{frcvr}}\ar@{=>}[d] &
    & \ar@{}[l]|{\mathsf{rcvr}}\circ\ar@{=>}[d] &
    & \circ\ar@{}[r]|{\mathsf{fabrt}}\ar@{=>}[d] &   \\
    \null & \ar[l]_{\AT}\bullet && \bullet\ar[r]^{\EOR} & \qquad
    \null & \ar[l]_{K\cons R}\bullet && \bullet\ar[r]^{\AT} & \null
  }\quad 
\xymatrix@C=6mm@R=5mm{
    & \null\ar[r]^{\CF} & \bullet\ar@{=>}[dl]\ar@{=>}[dr] &  &  \\
    & \ar@{}[l]|{\mathsf{rcvr}}\circ\ar@{=>}[d] &
    & \circ\ar@{}[r]|{\mathsf{fabrt}}\ar@{=>}[d] &   \\
    \null & \ar[l]_{\tagged{\CF}T\quad}\bullet && \bullet\ar[r]^{\AT} & \null
  }
  $$
  
  \caption{Trusted Third Party: 
    Abort (left), Resolve (center), and Confirm (right) Requests}
  \label{fig:wang:ttp}
\end{figure}
We have indicated here that $T$'s behavior, in response to an abort
request $\AR$ may lead either to an abort token $\AT$, or else to
evidence of receipt $\EOR$.  Now, the hollow nodes $\circ$
\emph{guard} the choice of branch.  $T$ transmits $\AR$ only after a
${\mathsf{abrt}}$ event, and $\EOR$ only after a ${\mathsf{frcvr}}$
event.  In response to a recovery request $\RR$ from $B$, $T$ may
transmit $K\cons R$ or an abort token $\AT$; however, the former
occurs only after a ${\mathsf{rcvr}}$ event and the latter only after
a ${\mathsf{fabrt}}$ event.  Thus, the essential job for $T$'s long
term state in this protocol is to ensure that if an ${\mathsf{abrt}}$
event occurs for session $L$, then a ${\mathsf{rcvr}}$ never happens
for $L$, and vice versa.  This is easily accomplished by a state-based
mechanism.  

\myparagraph{Dispute Resolution.}  A subtlety in this protocol
concerns dispute resolution.  
\label{point:dispute}
Since $A$ receives $\EOR$ before disclosing $K\cons R$, $A$ could
choose to abort at this point.  A dishonest $A$ could later choose
between proving delivery via $\EOR$ and proving that this session
aborted via the abort token $\AT$.  To prevent this, the protocol
stipulates that a judge resolving disputes queries $B$ or $T$ for an
abort token; it does not accept $A$'s presented $\EOR$ if the abort
token is also available.

However, this is asymmetric.  The abort token is used only by $B$ (or
$T$ on $B$'s behalf) to dispute receipt.  $A$ can never use it to
dispute origin~\cite[Sec.~4.4]{Wang2006}, because of essentially the
same abuse just mentioned.

For simplicity, we will assume that the judge is identical with $T$.
When asked by $A$ to confirm an $\EOR$, $T$ does so if the session has
not aborted.  When confirming an $\EOR$, $T$ must ensure that the
session will never abort in the future, so that an $\EOR$ confirmation
is handled similarly to a recovery request.  If the session has
already aborted, then $T$ returns the abort token instead.

This step may make Wang's protocol undesirable in some cases, where
$T$ may no longer be available for dispute resolution.  It is also why
Wang's protocol can use fewer messages than the four that
Pfitzmann-Schunter-Waidner proved to be needed in a fair exchange
protocol with asynchronous
communication~\cite{PfitzmannSchunterWaidner1998}.

\myparagraph{Our Correction to Wang's Protocol.}  We have adjusted
Wang's protocol.  When $B$'s recovery request arrives after $A$'s
abort request, $B$ receives $\tagged{\AR}T$.  In the original
description, $B$ receives $\AR$ itself.

However, then a dishonest $B$ has a strategy to defeat the fairness of
the protocol.  Namely, after receiving the first message, $B$ does not
reply to $A$, but immediately requests resolution from $T$, generally
receiving $K\cons R$ from $T$.  When $A$ requests an abort from $T$,
$B$ attempts to read this abort request off of the network.  If
successful, $B$ has both $\AR$ and $K\cons R$.  Hence, it can
subsequently choose whether to insist that the message was delivered,
using the valid $\EOO$, or whether to repudiate receipt, using the
$\AR$.

Whether this attack is possible depends on the nature of the channel
between $A$ and $T$.  Under the usual assumption that the channel is
resilient in the sense of ensuring delivery, the attack is possible.
If the channel offers both resilience and confidentiality, then the
attack would be impossible.  We have stipulated that $B$ needs the
countersigned $\tagged{\AR}T$ to make this attack infeasible on the
standard assumption of resiliency only. 

%
%

\myparagraph{Authentication Properties of Wang's Protocol.}  A
\emph{strand} is a (linearly ordered) sequence of nodes
$n_1\Rightarrow\ldots\Rightarrow n_j$, each of which represents
either:
\begin{quote}
  \begin{description}
    \item[Transmission] of some message $\term(n_i)$;
    \item[Reception] of some message $\term(n_i)$; or
    \item[State synchronization] labeled by some \emph{fact}, i.e.~a
    variable-free atomic formula, $E(a_1,\dots,a_k)$.
  \end{description}
\end{quote}
A strand may represent the behavior of a principal in a single local
session of a protocol, in which case it is a \emph{regular} strand of
that protocol, or it may represent a basic adversary activity.  Basic
adversary activities include receiving a plaintext and a key and
transmitting the result of the encryption, and receiving a ciphertext
and its matching decryption key, and transmitting the resulting
plaintext.  We show transmission and reception nodes by bullets
$\bullet$ and state synchronization nodes by hollow circles $\circ$.  

A \emph{protocol} $\Pi$ is a finite set of strands, which are the
\emph{roles} of the protocol.  A strand $s$ is an \emph{instance} of a
role $\rho\in\Pi$, if $s=\applyrep{\rho}\alpha$, i.e.~if $s$ results
from $\rho$ by applying a substitution $\alpha$ to parameters in
$\rho$.

A \emph{bundle} $\bnd$ is a finite directed acyclic graph whose
vertices are strand nodes, and whose arrows are either strand
edges$\Rightarrow$ or communication arrows $\rightarrow$.  A bundle
satisfies three properties:
\begin{enumerate}
  \item If $m\rightarrow n$, then $m$ is a transmission node, $n$ is a
  reception node, and $\term(m)=\term(n)$.
  \item Every reception node $n\in\bnd$ has exactly one incoming
  $\rightarrow$ arrow.
  \item If $n\in\bnd$ and $m\Rightarrow n$, then $m\in\bnd$.
\end{enumerate}
Bundles model possible protocol executions.  Bundles may include both
adversary strands and regular strands.  For more detail, see the
Appendix.

Using this notation, we can state two authentication properties that
involve $A,B$.  We omit a proof, which use digital signatures in an
extremely routine way, given a precise statement of the protocol.

\begin{lemma}\label{lemma:auth:init:resp}
  \begin{enumerate}
    \item Suppose $\bnd$ is a bundle in which $B$'s private signature
    key is uncompromised, and that, in $\bnd$, $A$ reaches a node
    marked ${\opr{depEOR}}$ on a strand with parameters $A,B,T,M,K,R$.
    Then $B$ has executed at least the first two nodes of a responder
    strand, transmitting $\EOR$, on a strand with matching parameters.
    \item Suppose $\bnd$ is a bundle in which $A$'s private signature
    key is uncompromised, and that, in $\bnd$, $B$ reaches a node
    marked ${\opr{depEOO}}$ or ${\opr{depAT}}$ on a strand with
    parameters $A,B,T,\EM,\EK$.  Then $A$ has executed at least the
    first node of an initiator strand, transmitting $\EOO$, on a
    strand with matching parameters.
  \end{enumerate}
\end{lemma}
Two authentication properties involving $T$ are also routine
applications of rules for digital signatures.
\begin{lemma}\label{lemma:auth:ttp}
  \begin{enumerate}
    \item Suppose $\bnd$ is a bundle in which $A$ and $T$'s private
    signature keys are uncompromised, and that, in $\bnd$, $A$ reaches
    a node marked ${\opr{depAT}}$.  Then $T$ has completed a strand
    transmitting $\AT$ with matching parameters.
    \item Suppose $\bnd$ is a bundle in which $A$ and $T$'s private
    signature keys are uncompromised.  If, in $\bnd$, $B$ reaches a
    node marked ${\opr{depAT}}$, then:
    \begin{enumerate}
      \item $A$ has reached the second node of an aborting strand,
      transmitting $\AR$, on a strand with matching
      parameters.\label{clause:B:AR:guarantee}
      \item $T$ has reached node transmitting $\AT$ in response to a
      recovery query $\RR$ with matching
      parameters.\label{clause:B:AT:guarantee}
    \end{enumerate}
    If instead $B$ reaches a node marked ${\opr{depEOO}}$ then either
    $A$ has transmitted $K\cons R$, or else $T$ has transmitted
    $K\cons R$.
  \end{enumerate}
\end{lemma}
Clause~(\ref{clause:B:AT:guarantee}) is the part of
Lemma~\ref{lemma:auth:ttp} that would be untrue without our adjustment
to Wang's protocol.  If $B$ receives only $\AR$, then
Clause~(\ref{clause:B:AR:guarantee}) holds, but not necessarily
Clause~(\ref{clause:B:AT:guarantee}).  This means that $T$'s state
might not reflect the abort.


\section{Protocol Behavior and Mutable State}
\label{sec:state}

We formalize state change using multiset
rewriting~\cite{MitchellEtAl99,DurginEtAl04}.  
%
%
Strands
contain special \emph{state synchronization events} that synchronize
them with the state of the principal executing the strands, as
formalized in Definition~\ref{defn:compatible}.

\subsection{Multiset rewriting to maintain state} 

We formalize mutable state using MSR.  A state is a multiset of ground
facts $F(t_1,\ldots,t_i)$, where each $F(t_1,\ldots,t_i)$ is the
application of a predicate $F$ to some sequence $t_1,\ldots,t_i$.
These arguments are messages, and thus do not contain variables;
hence, a state $\Sigma$ is a multiset of ground facts.  We write a
vector of messages $t,\ldots,t'$ in the form $\vec{t}$.

A rewrite rule $\rho$ takes the form:
$$D(\vec{t_0}),\ldots, F(\vec{t_1})
\stackrel{E(\vec{t_2})}{\longrightarrow}
G(\vec{t_3}), \ldots, H(\vec{t_4})$$ 
where now the arguments $\vec{t_0},\dots,\vec{t_3}$ are vectors of
parametric message terms that may contain variables.  When replacing
these variables with messages, we obtain ground facts.
Unlike~\cite{DurginEtAl04}, we label our transitions with a fact
${E(\vec{t_2})}$, but we do not require existential quantifiers in the
conclusions of rules.  We will assume that every variable free in
$\vec{t_0},\vec{t_1},\vec{t_3},\vec{t_4}$ is also free in $\vec{t_2}$.
Thus, a ground instance of ${E(\vec{t_2})}$ determines ground
instances of all the facts $D(\vec{t_0}),\ldots, F(\vec{t_1}),
G(\vec{t_3}), \ldots, H(\vec{t_4})$.

We write $\kind{lhs}(\rho)$ for $D(\vec{t_0}),\ldots, F(\vec{t_1})$;
we write $\kind{rhs}(\rho)$ for $G(\vec{t_3}), \ldots, H(\vec{t_4})$;
and $\kind{lab}(\rho)$ for ${E(\vec{t_2})}$.

A rule stipulates that the state can change by consuming instances of
the facts in its left-hand side, and producing the corresponding
instances of the facts in its right hand side.  These sets of facts
may overlap, in which case the facts in the overlap are required for
the rule to apply, but preserved when it executes.
A rewrite rule $\rho$ applies to a state $\Sigma_0$ when, for some
substitution $\sigma$,
$$\Sigma_0=\Sigma_0', D(\applysub{\vec{t_0}}{\sigma}),\ldots,
F(\applysub{\vec{t_1}}{\sigma}),$$
i.e., $\Sigma_0$ is the multiset union of $\Sigma_0'$ with instances
of the premises of $\rho$ under $\sigma$.  The result of applying
$\rho$ to $\Sigma_0$, using substitution $\sigma$, is
$$\Sigma_0', G(\applysub{\vec{t_3}}{\sigma}),\ldots,
H(\applysub{\vec{t_4}}{\sigma}).$$
Since this is a state, the facts
$G(\applysub{\vec{t_3}}{\sigma}),\ldots,
H(\applysub{\vec{t_4}}{\sigma})$ must again be ground; i.e.~$\sigma$
must associate the variables of $\vec{t_3},\ldots,\vec{t_4}$ with
variable-free messages.  There may be variables in
$\vec{t_3},\ldots,\vec{t_4}$ that do not occur in $\vec{t_0},\ldots,
{\vec{t_1}}$.  These variables take values nondeterministically, from
the point of view of the prior state.  In an execution, they may be
determined by protocol activities synchronized with the state.  Our
assumption about the variables in ${E(\vec{t_2})}$ ensures each ground
instance of ${E(\vec{t_2})}$ determines a $\sigma$ under which
$\vec{t_3},\ldots,\vec{t_4}$ become ground, and $t_2$ so to speak
summarizes all choices of values for variables.  
\begin{definition}
  Let $\rho=D(\vec{t_0}),\ldots, F(\vec{t_1})
  \stackrel{E(\vec{t_2})}{\longrightarrow} G(\vec{t_3}), \ldots,
  H(\vec{t_4})$.

  $\Sigma_0\stackrel{\rho,\sigma}{\longrightarrow}\Sigma_1$ a
  $\rho,\sigma$ transition from $\Sigma_0$ to $\Sigma_1$ iff
  $\Sigma_0,\Sigma_1$ are ground, and there exists a $\Sigma_0'$ such
  that $\Sigma_0=\Sigma_0', D(\applysub{\vec{t_0}}{\sigma}),\ldots,
  F(\applysub{\vec{t_1}}{\sigma})$ and
  $\Sigma_1=\Sigma_0',G(\applysub{\vec{t_3}}{\sigma}),\ldots,
  H(\applysub{\vec{t_4}}{\sigma})$.

  A \emph{computation} $\mathcal{C}$ is finite path through states via
  transitions; i.e.
  $\mathcal{C}=\Sigma_0\;\stackrel{\rho_0,\sigma_0}{\longrightarrow}\;
  \Sigma_1 \;\stackrel{\rho_1,\sigma_1}{\longrightarrow}\ldots
  \stackrel{\rho_j,\sigma_j}{\longrightarrow}\; \Sigma_{j+1}$.
  $\mathcal{C}$ is \emph{over} a set of rules $R$ if each $\rho_i\in
  R$.  When no ambiguity results, we will also write $\mathcal{C}$ in
  the form:
  $$\mathcal{C}=\Sigma_0\;
  \stackrel{E_0(\applysub{\vec{t_0}}{\sigma_0})}{\longrightarrow}\;\Sigma_1
  \;\stackrel{E_1(\applysub{\vec{t_1}}{\sigma_1})}{\longrightarrow}\ldots
  \stackrel{E_j(\applysub{\vec{t_j}}{\sigma_j})}{\longrightarrow}\;
  \Sigma_{j+1}.$$
  We write $\kind{first}(\mathcal{C})$ for $\Sigma_0$ and
  $\kind{last}(\mathcal{C})$ for $\Sigma_{j+1}$.  
\end{definition}
%
%
%
In this lemma, we interpret $\setminus,\cup,\subseteq$ as the multiset
difference, union, and subset operators.  
\begin{lemma}
  \label{lemma:local:concurrent} 
  Suppose $(\applysub{\kind{lhs}(\rho_1)}{\sigma_1}) \cup
  (\applysub{\kind{lhs}(\rho_2)}{\sigma_2}) \subseteq \Sigma_0$.
  If $\Sigma_0 \, \stackrel{\rho_1,\sigma_1}{\longrightarrow} \, \Sigma_1 \,
  \stackrel{\rho_2,\sigma_2}{\longrightarrow} \, \Sigma_2$, then
  $$\exists \Sigma_1'\qdot \Sigma_0 \,
  \stackrel{\rho_2,\sigma_2}{\longrightarrow} \, \Sigma_1' \,
  \stackrel{\rho_1,\sigma_1}{\longrightarrow} \, \Sigma_2.$$
\end{lemma}

\begin{proof}
  $\Sigma_1=(\Sigma_0 \setminus
  (\applysub{\kind{lhs}(\rho_1)}{\sigma_1})) \cup
  (\applysub{\kind{rhs}(\rho_1)}{\sigma_1})$, and $\Sigma_2=
  (\Sigma_1\setminus (\applysub{\kind{lhs}(\rho_2)}{\sigma_2})) \cup
  (\applysub{\kind{rhs}(\rho_2)}{\sigma_2})$.  We define
  $\Sigma_1'=(\Sigma_0 \setminus
  (\applysub{\kind{lhs}(\rho_2)}{\sigma_2})) \cup
  (\applysub{\kind{rhs}(\rho_2)}{\sigma_2})$.  By the assumption,
  $\Sigma_2 = (\Sigma_1'\setminus
  (\applysub{\kind{lhs}(\rho_1)}{\sigma_1})) \cup
  (\applysub{\kind{rhs}(\rho_1)}{\sigma_1})$.
\end{proof}

\subsection{Locality to principals} 


In our manner of using MSR, all manipulation of state is local to a
particular principal, and coordination among different principals
occurs only through protocol behavior represented on strands.
\begin{definition}
  A set of rewrite rules $R$ is \emph{localized to principals}, if,
  for a single distinguished variable $p$, for every rule $\rho\in R$,
  for each fact $F(\vec{t})$ occurring in $\rho$ as a premise or
  conclusion, $F(\vec{t})$ is of the form $F(p,\vec{t'})$.

  The \emph{principal of} a transition
  $\Sigma_0\stackrel{\rho,\sigma}{\longrightarrow}\Sigma_1$ is
  $\applysub{p}{\sigma}$.  
\end{definition}
Thus, only the principal of a transition
$\Sigma_0\stackrel{\rho,\sigma}{\longrightarrow}\Sigma_1$ is affected
by it.  Transitions with different principals are always concurrent.
If $\applysub{p}{\sigma_1}\not=\applysub{p}{\sigma_2}$ and
$({\rho_1,\sigma_1}), ({\rho_2,\sigma_2})$ can happen, so can the
reverse, with the same effect:
\begin{cor}
  \label{prop:local:concurrent}
  Let $R$ be localized to principals, with $\rho_1,\rho_2\in R$, and
  $\applysub{p}{\sigma_1}\not=\applysub{p}{\sigma_2}$.  If $\Sigma_0
  \, \stackrel{\rho_1,\sigma_1}{\longrightarrow} \, \Sigma_1 \,
  \stackrel{\rho_2,\sigma_2}{\longrightarrow} \, \Sigma_2$, then
  $\Sigma_0 \, \stackrel{\rho_2,\sigma_2}{\longrightarrow} \,
  \Sigma_1'\, \stackrel{\rho_1,\sigma_1}{\longrightarrow} \,
  \Sigma_2$, for some $\Sigma_1'$.
\end{cor}
\begin{proof}
  Since $\applysub{p}{\sigma_1}\not=\applysub{p}{\sigma_2}$, the facts
  on the right hand side of $\applysub{\rho_1}{\sigma_1}$ are disjoint
  from those on the left hand side of $\applysub{\rho_2}{\sigma_2}$.
  Hence, ${\rho_2,\sigma_2}$ being enabled in $\Sigma_1$, it must also
  be enabled in $\Sigma_0$.  Hence,
  $(\applysub{\kind{lhs}(\rho_1)}{\sigma_1}) \cup
  (\applysub{\kind{lhs}(\rho_2)}{\sigma_2})\subset\Sigma_0$, and we
  may apply Lemma~\ref{lemma:local:concurrent}.  
\end{proof}
The following definition connects bundles with computations.
\begin{definition}
  \label{defn:compatible}
  Let $R$ be localized to principals.

  \begin{enumerate}
  \item An \emph{eventful protocol} $\Pi$ is a finite set of roles
    containing nodes of three kinds:
    \begin{enumerate}
      \item transmission nodes $+t$, where $t$ is a message;
      \item reception nodes $-t$, where $t$ is a message; and
      \item state synchronization events $E_i(p,\vec{t})$.
    \end{enumerate}
    We require that if $E_i(p,\vec{t})$ and $E_j(p',\vec{t'})$ lie on
    the same strand, then $p=p'$.  If a strand $s$ contains a state
    synchronization $E_i(p,\vec{t})$, then $p$ is \emph{the principal
      of} $s$.
    \item Suppose that $\bnd$ is a bundle over the eventful protocol
    $\Pi$; $\mathcal{C}$ is a finite computation for the rules $R$;
    and $\phi$ is a bijection between state synchronization nodes of
    $\bnd$ and transitions $E_i(\vec{t_i})$ of $\mathcal{C}$.  $\bnd$
    is \emph{compatible with} $\mathcal{C}$ \emph{under} $\phi$ iff
    \begin{enumerate}
      \item The event $E_i(p,\vec{t})$ at $n$ is the label on
      $\phi(n)$, and
      \item $n_0\preceq_{\bnd} n_1$ implies $\phi(n_0)$ precedes
      $\phi(n_1)$ in $\mathcal{C}$.
    \end{enumerate}
    \item An \emph{execution} of $\Pi$ \emph{constrained by} $R$ is a
    triple $(\bnd,\mathcal{C},\phi)$ where $\bnd$ is compatible
    with $\mathcal{C}$ under $\phi$.  
  \end{enumerate}
\end{definition}
If $(\bnd,\mathcal{C},\phi)$ is an execution, then it represents
possible protocol behavior $\bnd$ for $\Pi$, where state-sensitive
steps are constrained by the state maintained in $\mathcal{C}$.
Moreover, the state $\mathcal{C}$ evolves as driven by state
synchronizations occurring in strands appearing in $\bnd$.  The
bijection $\phi$ makes explicit the correlation between events in the
protocol runs of $\bnd$ and transitions occurring in $\mathcal{C}$.

\subsection{States and Rules for Wang's Protocol}
\label{sec:state:wang}

\myparagraph{Trusted Third Party State.} Conceptually, the trusted
third party $T_0$ maintains a status record for each possible
transaction it could be asked to abort or recover.  Since each
transaction is determined by a label
$\mathcal{L}_m(\vr{hm},\vr{hk})=A\cons B\cons
T\cons\vr{hm}\cons\vr{hk}$, where $T=T_0$, it maintains a fact for
each such value.  This fact indicates either (1) that the no message
has as yet been received in connection with this session; or (2) that
the session has been recovered, in which case the evidence of receipt
is also kept in the record; or (3) that the session has been aborted,
in which case the signed abort request is also kept in the record.
Thus, the state record for the session with label
$\ell=\mathcal{L}_m(\vr{hm},\vr{hk})$ is a fact of one of the three
forms:
$$  
\opr{unseen}(T,\ell) \qquad \opr{recovered}(T,\ell, \EOR) \qquad
\opr{aborted}(T,\ell,{\AT})
$$
Naturally, a programmer will maintain a sparse representation of this
state, in which only the last two forms are actually stored.  A query
for $\ell$ that retrieves nothing indicates that the session $\ell$ is
as yet unseen.

Four types of events synchronize with $T$'s state.  The event
$\opr{rcvr}(\ell,e)$ deposits a $\opr{recovered}(\ell, e)$ fact into
the state, and requires the state to contain either an
$\opr{unseen}(\ell)$ fact or a preexisting $\opr{recovered}(\ell, e)$
fact with the same $e$, which are consumed.
$$  
\xymatrix@R=2mm{
  { \opr{unseen}(T,\ell) }\ar[rr]^{\opr{rcvr}(T,\ell,e)} & &
  {\opr{recovered}(T,\ell,e) } \\
  { \opr{recovered}(T,\ell,e) }\ar[rr]^{\opr{rcvr}(T,\ell,e)} & & 
  { \opr{recovered}(T,\ell,e) }
}
$$
The second of these forms ensures that repeated $\opr{rcvr}$ events
succeed, with no further state change.  

The event $\opr{abrt}(T,\ell,a)$ deposits a $\opr{aborted}(T,\ell,a)$ fact
into the state, and requires the state to contain either an
$\opr{unseen}(T,\ell)$ fact or a preexisting $\opr{aborted}(T,\ell,a)$
fact, which are consumed.
$$  
\xymatrix@R=2mm{
  { \opr{unseen}(T,\ell) }\ar[rr]^{\opr{abrt}(T,\ell,a)} & &
  { \opr{aborted}(T,\ell,a) } \\
  { \opr{aborted}(T,\ell,a) }\ar[rr]^{\opr{abrt}(T,\ell,a)} & & 
  { \opr{aborted}(T,\ell,a) }
}
$$
Finally, there is an event for a forced recover
${\opr{frcvr}(T,\ell,e)}$ and one for a forced abort
${\opr{fabrt}(T,\ell,a)}$.  These may occur when the $\opr{recovered}$
fact [or respectively, the $\opr{aborted}$ fact] is already present,
so that attempt to abort [or respectively, to recover] must yield the
opposite result.
$$  
\xymatrix@R=2mm{
  { \opr{recovered}(T,\ell,e) }\ar[rr]^{\opr{frcvr}(T,\ell,e)} & &
  { \opr{recovered}(T,\ell,e) } \\
  { \opr{aborted}(T,\ell,a) }\ar[rr]^{\opr{fabrt}(T,\ell,a)} & & 
  { \opr{aborted}(T,\ell,a) }
}
$$

\begin{definition}
\label{def:wang:initial}
A \emph{$\kind{GW}$ initial state} is a multiset $\Sigma$ such that:
\begin{enumerate}
  \item No fact $\opr{recovered}(T,\ell,e)$ or $\opr{aborted}(T,\ell,a)$
  is present in $\Sigma$;
  \item For all $\ell$, the multiplicity
  $\multiplicity{\opr{unseen}(T,\ell)}{\Sigma}$ of $\opr{unseen}(T,\ell)$
  in $\Sigma$ is at most 1.
\end{enumerate}  
$\mathcal{C}$ is a \emph{$\kind{GW}$ computation} if it is a
computation using the set $R_W$ of the six rules above, starting from
a $\kind{GW}$ initial state $\Sigma_0$.  
\end{definition}
There are several obvious consequences of the definitions.  The first
says that the multiplicity of facts for a single session $\ell$ does
not increase, and initially starts at 0 or 1, concentrated in
$\opr{unseen}(T,\ell)$.  The next two say that a
${\opr{recovered}(T,\ell,e)}$ fact arises only after a
${\opr{rcvr}(T,\ell,e)}$ event, and a ${\opr{aborted}(T,\ell,a)}$ fact
after an ${\opr{abrt}(T,\ell,e)}$ event.  Then we point out that a
${\opr{rcvr}(T,\ell,e)}$ event and an ${\opr{abrt}(T,\ell,a)}$ event
never occur in the same computation, and finally that a
${\opr{rcvr}(T,\ell,e)}$ event must precede a
${\opr{frcvr}(T,\ell,e)}$ event, and likewise for aborts and forced
aborts.
\begin{lemma}
\label{lemma:GW:computation}
  Let
  $\mathcal{C}=\Sigma_0\;\stackrel{\rho_0,\sigma_0}{\longrightarrow}\;
  \Sigma_1 \;\stackrel{\rho_1,\sigma_1}{\longrightarrow}\ldots
  \stackrel{\rho_j,\sigma_j}{\longrightarrow}\; \Sigma_{j+1}$ be a
  $\kind{GW}$ computation.
  \begin{enumerate}
    \item For any $\ell$ and $i\le j+1$, the sum over all $e,a$ of the
    multiplicities of all facts $\opr{unseen}(T,\ell)$,
    $\opr{recovered}(T,\ell,e)$, $\opr{aborted}(T,\ell,a)$ is unchanged: 
    \begin{eqnarray*} 
      1\ge\multiplicity{\opr{unseen}(T,\ell)}{\Sigma_0} & = & 
       \sum_{a,e}\big(\,\multiplicity{\opr{unseen}(T,\ell)}{\Sigma_i} +
       \multiplicity{\opr{recovered}(T,\ell,e)}{\Sigma_i} \\
       & + & \multiplicity{\opr{aborted}(T,\ell,a)}{\Sigma_i}\,\big).
    \end{eqnarray*}
    \item $\multiplicity{\opr{recovered}(T,\ell,e)}{\Sigma_i}=1$ iff
    $\exists k<i$,
    $\applysub{\kind{lab}(\rho_k)}{\sigma_k}={\opr{rcvr}(T,\ell,e)}$.
    \item $\multiplicity{\opr{aborted}(T,\ell,a)}{\Sigma_i}=1$ iff
    $\exists k<i$,
    $\applysub{\kind{lab}(\rho_k)}{\sigma_k}={\opr{abrt}(T,\ell,e)}$.
    \item If $\exists i$,
    $\applysub{\kind{lab}(\rho_i)}{\sigma_i}={\opr{rcvr}(T,\ell,e)}$,
    then $\forall k,a$,
    $\applysub{\kind{lab}(\rho_k)}{\sigma_k}\not={\opr{abrt}(T,\ell,a)}$.
    \label{clause:rcvr:abort:exclusive}
    \item If $\exists i$,
    $\applysub{\kind{lab}(\rho_i)}{\sigma_i}={\opr{frcvr}(T,\ell,e)}$,
    then $\exists k<i$,
    $\applysub{\kind{lab}(\rho_k)}{\sigma_k}={\opr{rcvr}(T,\ell,e)}$.
    \item If $\exists i$,
    $\applysub{\kind{lab}(\rho_i)}{\sigma_i}={\opr{fabrt}(T,\ell,a)}$,
    then $\exists k<i$,
    $\applysub{\kind{lab}(\rho_k)}{\sigma_k}={\opr{abrt}(T,\ell,a)}$.  
    \item If $\opr{unseen}(T,\ell)\in\Sigma_0$, then every session
    $\ell$ request to $T$ in Fig.~\ref{fig:wang:ttp} can proceed on
    some branch.  \label{clause:T:progresses}
  \end{enumerate}
\end{lemma}
%

\myparagraph{Initiator and Responder State.}  The initiator and
responder have rules with empty precondition, that simply deposit
records values into their state.  These records are of the forms
$\kind{eor}(A,\ell,\EOR,M,K,R)$, $\kind{eoo}(B,\ell,\EOO,M,K,R)$, and
$\kind{aborted}(P,\ell,\tagged{\tagged{\opr{ab\_rq}\cons
    \h(\ell)}A}T)$.  The last is used both by the initiator and the
responder.  The rules are:
$$
\xymatrix@R=2mm{
  { \cdot\qquad }\ar[rrrr]^{{\mathsf{depEOR}(A,\ell,e,M,K,R)}} & & & & {
    \kind{eor}(A,\ell,e,M,K,R) }
  \\
  { \cdot\qquad }\ar[rrrr]^{{\mathsf{depEOO}}(B,\ell,e,M,K,R)} & & & & {
    \kind{eoo}(B,\ell,e,M,K,R) }
  \\
  { \cdot\qquad }\ar[rrrr]^{{\mathsf{depAT}}(P,\ell,a)} & & & & {
    \kind{aborted}(P,\ell,a) }}
$$


\section{Progress Assumptions}
\label{sec:progress}

We introduce two kinds of progress properties for protocols.  One of
them (Def.~\ref{defn:guar:del}) formalizes the idea that certain
messages, if sent, must be delivered to a regular participant,
i.e.~that these messages traverse resilient channels.  The second is
the idea that principals, at particular nodes in a strand, must
progress.  We will stipulate that a principal whose next step is a
state event, and the current state satisfies the right hand side of
the associated rule, then the principal will always take either that
step or another enabled step.  It is formalized in
Def.~\ref{def:stable}.   

\begin{definition}
  \label{defn:guar:del}
  Suppose that $\Pi$ is a protocol, and $G$ is a set of nodes
  $\strandnode{s}{i}$ such that for all $\strandnode{s}{i}\in G$, $s$
  is a role of $\Pi$ and $\strandnode{s}{i}$ is a transmission node.
  Then $G$ is a set of \emph{guaranteed delivery assumptions} for
  $\Pi$.

  A transmission node $n$ on a strand $s'$ is a \emph{guaranteed
    delivery node} for $\Pi,G$ if it is an instance
  $n=\applysub{(\strandnode{s}{i})}\alpha$ of the $i^{\mathrm{th}}$
  node of some role $s\in\Pi$, and $\strandnode{s}{i}\in G$.

  Let $\bnd$ be a bundle for $\Pi$.  $\bnd$ \emph{satisfies guaranteed
    delivery} for $G$ if, for every guaranteed delivery node
  $n\in\bnd$, there is a unique node $m\in\bnd$ such that
  $n\rightarrow_{\bnd}m$, and moreover $m$ is regular.   
\end{definition}

There are three ingredients here.  First, $n$'s transmission should be
received somewhere.  Second, it should be received at most once.
Finally, the recipient should be regular.  For our progress condition,
however, we want a stronger condition than this guaranteed delivery
property.  In particular, we also want to stipulate that if a
guaranteed transmission node can be added, and its message can be
delivered, then it will be added together with one matching reception
node. However, for this we need to define the right notion of ``can.'' 
%
%
%
Thus, we define the unresolved nodes of a bundle, using $n\sim m$,
which means that $n$ and $m$ are similar in the following sense:
\begin{definition}
  \label{def:similar}
  Regular nodes $n',m'$ are \emph{similar}, written $n'\sim m'$, if
  the initial segments of the strands they lie on, 
  $n\Rightarrow\ldots\Rightarrow n'$ and
  $m\Rightarrow\ldots\Rightarrow m'$,
  (1) are of the same length; (2) corresponding nodes have the same
  direction (transmission, reception, or state synchronization); and
  (3) corresponding nodes have the same message or state
  synchronization event label.

  A regular node $n_0$ is \emph{unresolved in} $\bnd$ if
  $n_0\Rightarrow n_1$ and for some $n_0'\in\bnd$, $n_0'\sim n_0$ but
  for all $n_1'\in\bnd$, $n_0'\not\Rightarrow n_1'$.  
\end{definition}
A node $n_0$ is unresolved if it \emph{can} progress to some $n_1$,
but a similar $n_0'\in\bnd$ \emph{has not} progressed.  Thus,
substituting a similar node for a node in $\bnd$, we obtain a bundle
$\bnd'$ to which this transition may be added.
\begin{definition}
  \label{def:stable}
  Let $\mathcal{E}=(\bnd,\mathcal{C},\phi)$ be an execution of $\Pi,G$
  constrained by $R$.  $\mathcal{E}$ is a \emph{stable execution} if
  (1) $\bnd$ satisfies guaranteed delivery for $G$; (2) there are no
  enabled transmission edges for $\bnd$; and (3) there are no enabled
  state edges for $\mathcal{E}$, where we define \emph{enabled
    transmission and state edges} as follows:
  \begin{enumerate}
    \item $n_0\Rightarrow n_1$ is an \emph{enabled transmission edge
      for} $\mathcal{E}$ if:
    \begin{enumerate}
    \item $n_0$ is unresolved in $\bnd$; 
    \item $n_1$ is a guaranteed delivery node; and 
    \item there is a regular reception node $m_1$ with
      $\term(m_1)=\term(n_1)$ where either
      \begin{enumerate}
        \item $m_1$ is the first node on its strand, or else
        \item $m_0\Rightarrow m_1$, where $m_0$ is unresolved in
        $\bnd$.
      \end{enumerate}
    \end{enumerate}
  \item $n_0\Rightarrow n_1$ is an \emph{enabled state edge for}
    $\mathcal{E}$ if:
    \begin{enumerate}
    \item $n_0$ is unresolved in $\bnd$;
    \item $n_1$ is a state synchronization node with event $E(p,\vec{t})$; and
    \item $\exists \rho\in R$ and $\sigma$
      s.t. $\applysub{\kind{lab}(\rho)}\sigma=E(p,\vec{t})$ and
      $\applysub{\kind{lhs}(\rho)}{\sigma}\subseteq\kind{last}(\mathcal{C})$.
    \end{enumerate}
  \end{enumerate}
\end{definition}
In a stable execution, each strand has reached a ``stopping point,''
where no transmission with guaranteed delivery (and matching
reception) is waiting to happen, and no state synchronization event is
waiting to happen.
A protocol $\Pi$ and rules $R$ drive the evolution of state through
states satisfying some balance property $\Psi$ means that when
$\mathcal{E}=(\bnd,\mathcal{C},\phi)$ is a stable execution for
$\Pi,R$, and $\Psi(\kind{first}(\mathcal{C}))$, then
$\Psi(\kind{last}(\mathcal{C}))$.

\myparagraph{Guaranteed Delivery for Wang's Protocol.}  The guaranteed
delivery assumptions for Wang's protocol are not surprising.  They are
the messages transmitted on resilient channels between the principals
and the Trusted Third Party.  These are $A$'s transmission of $\AR$
and $B$'s transmission of $\RR$ in Fig~\ref{fig:wang:init:resp}, and
$T$'s six transmissions in Fig.~\ref{fig:wang:ttp}.

\myparagraph{Progress in Wang's Protocol.}  No protocol can protect
principals that do not follow it.  Thus, correctness conditions are
stated for stable executions in which at least one of $A,B$ comply
with the protocol.  We also assume that the trusted third party $T$
merits trust, and also complies with the protocol.  A principal $P$ is
\emph{compliant} in a bundle $\bnd$ if $P\in\{A,B\}$ and $P$'s signing
key is used only in accordance with $\Pi_{GW}$ in $\bnd$; or if $P=T$,
the trusted third party, and $T$'s signing and decryption keys are
used only in accordance with $\Pi_{GW}$ in $\bnd$.

Henceforth, let $\mathcal{E}=(\bnd,\mathcal{C},\phi)$ be a
$\opr{GW}$-execution.  Let $\Sigma_0$ and $\Sigma_j$ be the first and
last states of $\mathcal{C}$.  For each label $\ell$ occurring in an
$A$ initiator strand or a $B$ responder strand in $\bnd$, assume that
${\opr{unseen}(T,\ell)}\in\Sigma_0$.

\begin{lemma}[$\opr{GW}$ Progress]
  \label{thm:progress}
  Let $S$ be a set of principals compliant in $\mathcal{E}$, with
  $T\in S$.  There exists a stable
  $\mathcal{E}'=(\bnd',\mathcal{C}',\phi')$, such that (1)
  $\mathcal{E}'$ extends $\mathcal{E}$; (2) the principals $S$ are
  compliant in $\mathcal{E}'$; and (3) $p=T$ if $p$ is the principal
  of any regular strand of $\bnd'$ that does not appear in $\bnd$.

  If $s$ is an initiator or TTP strand with $\bnd'$-height $\ge 1$,
  then its $\bnd'$-height is its full length.  If $s$ is a responder
  strand with $\bnd'$-height $\ge 2$, then its $\bnd'$-height is its
  full length.
\end{lemma}
\begin{proof}
  Inspecting Fig.~\ref{fig:wang:init:resp}, we see that an initiator
  strand of $\bnd$-height 1 may progress by sending a
  guaranteed-delivery $\AR$, which is also possible for an initiator
  strand that has received $\EOR$.  The guaranteed delivery rule
  requires the first node of some $T$ strand receiving $\AR$.  By
  Lemma~\ref{lemma:GW:computation}, Clause~\ref{clause:T:progresses},
  some $T$ state synchronization event is enabled, after which $T$
  makes a guaranteed-delivery transmission.  Thus, $A$ receives $\AT$
  or $\EOR$.  Since its deposit state synchronization events have
  empty precondition, $A$ will complete its strand.
  The analysis for responder strands is similar.  
\end{proof}
That is, we may regard starting a strand in $\bnd$, or---for a
responder---sending its $\EOR$ message, as a promise to progress
regularly in the future, as required by Def.~\ref{def:stable}.
Moreover, new strands that begin in $\bnd'$, not $\bnd$, belong only
to the TTP $T$.  In $\bnd'$, these strands have terminated by reaching
its full length.


\section{Correctness of Wang's protocol}
\label{sec:correctness}

We now summarize our conclusions in a theorem that puts together the
different elements we have discussed.  
\begin{thm}\label{thm:b:completes}
  Let $\mathcal{E}=(\bnd,\mathcal{C},\phi)$ be a \emph{stable}
  $\opr{GW}$-execution with $\opr{unseen}(T,\ell)\in\Sigma_0$.  
  \begin{enumerate}
    \item If $\kind{eoo}(B,\ell,\EOO,M,K,R)\in\Sigma_j$ but
    $\not\in\Sigma_0$, then for compliant $A$,
    $\kind{eor}(A,\ell,\EOR,M,K,R)\in\Sigma_j$.
  \item If $\kind{eor}(A,\ell,\EOR,M,K,R)\in\Sigma_j$ but
    $\not\in\Sigma_0$, then for compliant $B$, either
    $\kind{eoo}(B,\ell,\EOO,M,K,R)\in\Sigma_j$ or else
    $\kind{aborted}(B,\ell,\AT)\in\Sigma_j$.  

  \end{enumerate}
 \end{thm}
 \begin{proof}
   1.  By the state rules for $B$, $\kind{depEOO}(B,\ell,e,M,K,R)$ has
   occurred in $\mathcal{C}$.  Hence, $B$ has reached one of the two
   $\kind{depEOO}()$ nodes shown in Fig~\ref{fig:wang:init:resp}, with
   parameters $B,\ell,e,M,K,R$.  Hence, by
   Lemma~\ref{lemma:auth:init:resp}, Clause~2, $A$ has executed at
   least the first node of an initiator strand, transmitting $\EOO$,
   on a strand with matching parameters.  Since $\mathcal{E}$ is
   stable, by Thm.~\ref{thm:progress}, $A$'s strand has full height.
   Thus, either $\kind{depEOR}()$ or $\kind{depAT}()$ has occurred
   with matching parameters.  

   However, if $\kind{depAT}()$ has occurred at $A$, then $A$ does not
   transmit $K\cons R$.  Moreover, since $A$ has received $\AT$, $T$
   has transmitted $\AT$ by Lemma~\ref{lemma:auth:ttp}, Clause~1.
   Hence, by Lemma~\ref{lemma:GW:computation},
   Clause~\ref{clause:rcvr:abort:exclusive}, $\mathcal{C}$ does not
   contain a $\opr{rcvr}(T,\ell,e)$ event.  Thus, contrary to
   Lemma~\ref{lemma:auth:ttp}, Clause~2, $T$ has not transmitted
   $K\cons R$.  Hence, $\kind{depEOR}()$ has occurred.

   2.  By the state rules for $A$, $\kind{depEOR}(A,\ell,e,M,K,R)$ has
   occurred in $\mathcal{C}$.  Hence, $A$ has reached one of the two
   $\kind{depEOR}()$ nodes shown in Fig~\ref{fig:wang:init:resp}, with
   parameters $A,\ell,e,M,K,R$.  Hence, by
   Lemma~\ref{lemma:auth:init:resp}, Clause~1, $B$ has executed at
   least the first two nodes of a responder strand, transmitting
   $\EOR$, on a strand with matching parameters.  Since $\mathcal{E}$
   is stable, by Thm.~\ref{thm:progress}, $B$'s strand has full
   height.  Thus, either $\kind{depEOO}()$ or $\kind{depAT}()$ has
   occurred at $B$ with matching parameters.
%
%
 \end{proof}
\myparagraph{Conclusion.}  This formalism has also been found to be
convenient to model the interface to a cryptographic device, the
Trusted Platform Module, which combines cryptographic operations with
a repository of state.  Thus, it appears to be  a widely applicable
approach to the problem of combining reasoning about cryptographic
protocols with reasoning about state and histories.



{
  \bibliography{secureprotocols.bib}

\begin{thebibliography}{10}

\bibitem{AsokanEtAl00}
N.~Asokan, Victor Shoup, and Michael Waidner.
\newblock Optimistic fair exchange of digital signatures.
\newblock {\em {IEEE} J. Sel.~Areas in Comms.}, 18(4):593--610, 2000.

\bibitem{CederquistEtAl07}
Jan Cederquist, Mohammad~Torabi Dashti, and Sjouke Mauw.
\newblock A certified email protocol using key chains.
\newblock In {\em Advanced Information Networking and Applications
  Workshops/Symposia {(AINA'07)}, Symposium on Security in Networks and
  Distributed Systems {(SSNDS07)}}, volume~1, pages 525--530. IEEE CS Press,
  2007.

\bibitem{MitchellEtAl99}
I.~Cervesato, N.~A. Durgin, P.~D. Lincoln, J.~C. Mitchell, and A.~Scedrov.
\newblock A meta-notation for protocol analysis.
\newblock In {\em Proceedings, 12th {IEEE} Computer Security Foundations
  Workshop}. {IEEE} Computer Society Press, June 1999.

\bibitem{ChadhaEtAl03}
Rohit Chadha, John~C. Mitchell, Andre Scedrov, and Vitaly Shmatikov.
\newblock Contract signing, optimism, and advantage.
\newblock In {\em Concur --- Concurrency Theory}, LNCS, pages 366--382.
  Springer, 2003.

\bibitem{TorabiDashti07}
Mohammad~Torabi Dashti.
\newblock {\em Keeping Fairness Alive}.
\newblock PhD thesis, Vrije Universiteit, Amsterdam, 2007.

\bibitem{DoghmiGuttmanThayer07}
Shaddin~F. Doghmi, Joshua~D. Guttman, and F.~Javier Thayer.
\newblock Searching for shapes in cryptographic protocols.
\newblock In {\em Tools and Algorithms for Construction and Analysis of Systems
  {(TACAS)}}, number 4424 in LNCS, pages 523--538. Springer, March 2007.
\newblock Extended version at {URL}:\url{http://eprint.iacr.org/2006/435}.

\bibitem{DurginEtAl04}
Nancy Durgin, Patrick Lincoln, John Mitchell, and Andre Scedrov.
\newblock Multiset rewriting and the complexity of bounded security protocols.
\newblock {\em Journal of Computer Security}, 12(2):247--311, 2004.
\newblock Initial version appeared in \emph{Workshop on Formal Methods and
  Security Protocols}, 1999.

\bibitem{EvenYacobi80}
Shimon Even and Yacov Yacobi.
\newblock Relations among public key signature systems.
\newblock Technical Report 175, Computer Science Departament, Technion, 1980.

\bibitem{Guttman09}
Joshua~D. Guttman.
\newblock Cryptographic protocol composition via the authentication tests.
\newblock In Luca de~Alfaro, editor, {\em Foundations of Software Science and
  Computation Structures ({FOSSACS})}, number 5504 in LNCS, pages 303--317.
  Springer, March 2009.

\bibitem{GuttmanThayer02}
Joshua~D. Guttman and F.~Javier {{Thayer}}.
\newblock Authentication tests and the structure of bundles.
\newblock {\em Theoretical Computer Science}, 283(2):333--380, June 2002.
\newblock Conference version appeared in \emph{IEEE Symposium on Security and
  Privacy}, May 2000.

\bibitem{PfitzmannSchunterWaidner1998}
Birgit Pfitzmann, Matthias Schunter, and Michael Waidner.
\newblock Optimal efficiency of optimistic contract signing.
\newblock In {\em Seventeenth Annual {ACM} Symposium on Principles of
  Distributed Computing}, pages 113--122, New York, May 1998. ACM.

\bibitem{Rabin81}
Michael Rabin.
\newblock How to exchange secrets by oblivious transfer.
\newblock Technical report, Technical Report TR-81, Harvard Aiken Computation
  Laboratory, 1981.
\newblock Available at \url{http://eprint.iacr.org/2005/187}.

\bibitem{Wang2006}
Guilin Wang.
\newblock Generic non-repudiation protocols supporting transparent off-line
  {TTP}.
\newblock {\em Journal of Computer Security}, 14(5):441--467, 2006.

\end{thebibliography}
  \ifieee
  \bibliographystyle{latex8}               
  \else 
  \bibliographystyle{plain}               
  \fi
}

\appendix
\section{Messages and Protocols}
\label{sec:strands} 

In this appendix, we provide an overview of the current strand space
framework; this section is essentially identical with part
of~\cite{Guttman09}.  

\paragraph{Message Algebra.}  Let $\Algebra_0$ be an algebra equipped
with some operators and a set of homomorphisms
$\eta\colon\Algebra_0\rightarrow\Algebra_0$.  We call members of
$\Algebra_0$ \emph{atoms}.

For the sake of definiteness, we will assume here that $\Algebra_0$ is
the disjoint union of infinite sets of \emph{nonces}, \emph{atomic
  keys}, \emph{names}, and \emph{texts}.  The operator $\kind{sk}(a)$
maps names to (atomic) signature keys, and $K^{-1}$ maps an asymmetric
atomic key to its inverse, and a symmetric atomic key to itself.
Homomorphisms $\eta$ are maps that respect sorts, and act
homomorphically on $\kind{sk}(a)$ and $K^{-1}$.

Let $X$ is an infinite set disjoint from $\Algebra_0$; its
members---called \emph{indeterminates}---act like unsorted variables.
$\Algebra$ is freely generated from $\Algebra_0\cup X$ by two
operations: encryption $\enc{t_0}{t_1}$ and tagged concatenation
$\tagname{tag} t_0\cons t_1$, where the tags $\tagname{tag}$ are drawn
from some set $\mathit{TAG}$.  For a distinguished tag
$\tagname{nil}\!\!$, we write $\tagname{nil}\; t_0\cons t_1$ as
$t_0\cons t_1$ with no tag.  In $\enc{t_0}{t_1}$, a non-atomic key
$t_1$ is a symmetric key.  Members of $\Algebra$ are called
\emph{messages}.

A homomorphism $\alpha=(\eta,\chi)\colon\Algebra\rightarrow\Algebra$
consists of a homomorphism $\eta$ on atoms and a function $\chi\colon
X\rightarrow\Algebra$.  It is defined for all $t\in\Algebra$ by the
conditions:
\begin{quote}
  \begin{tabular}{r@{$\;=\;$}l@{\quad}l@{\qquad\qquad}r@{$\;=\;$}l}
    $\applyrep{a}{\alpha}$ & $\eta(a)$, & if $a\in\Algebra_0$ &
    $\applyrep{\enc{t_0}{t_1}}{\alpha}$ &
    $\enc{\applyrep{t_0}{\alpha}}{\applyrep{t_1}{\alpha}}$    \\
    $\applyrep{x}{\alpha}$ & $\chi(x)$, & if $x \in X$ & 
    $\applyrep{\tagname{tag} t_0\cons t_1}{\alpha}$ & 
    $\tagname{tag} \applyrep{t_0}{\alpha}\cons
    \applyrep{t_1}{\alpha}$
\end{tabular}
\end{quote}
Thus, atoms serve as typed variables, replaceable only by other values
of the same sort, while indeterminates $x$ are untyped.
Indeterminates $x$ serve as blank slots, to be filled by any
$\chi(x)\in\Algebra$.  Indeterminates and atoms are jointly
\emph{parameters}.

Messages are abstract syntax trees in the usual way:
\begin{enumerate}
\item Let $\ell$ and $r$ be the partial functions such that for
  $t=\enc{t_1}{t_2}$ or $t=\tagname{tag}{t_1}\cons{t_2}$,
  $\ell(t)=t_1$ and $r(t)=t_2$; and for $t\in\Algebra_0$, $\ell$ and
  $r$ are undefined.
  \item A \emph{path} $p$ is a sequence in $\{\ell,r\}^{*}$.  We
  regard $p$ as a partial function, where $\seq{}=\kind{Id}$ and
  $\kind{cons}(f,p)=p\circ f$.  When the rhs is defined, we have:  1.
  $\seq{}(t)=t$; 2.  $\kind{cons}({\ell},p)(t)=p(\ell(t))$; and 3.
  $\kind{cons}({r},p)(t)=p(r(t))$.
  \item $p$ \emph{traverses a key edge} in $t$ if $\pleads{t}{p_1}$ is
  an encryption, where $p=p_1\supfrown \seq{r}\supfrown p_2$.
  \item $p$ \emph{traverses a member of} $S$ if $\pleads{t}{p_1}\in
  S$, where $p=p_1\supfrown p_2$ and $p_2\not=\seq{}$.
  \item $t_0$ \emph{is an ingredient of} $t$, written $t_0\subterm t$,
  if $t_0=\pleads{t}{p}$ for some $p$ that does not traverse a key
  edge in $t$.%
%
%
\item $t_0$ \emph{appears in} $t$, written $t_0\ll t$, if
  $t_0=\pleads{t}{p}$ for some $p$.
\end{enumerate}
A single local session of a protocol at a single principal is a
\emph{strand}, containing a linearly ordered sequence of
transmissions, receptions, and state synchronization events that we
call \emph{nodes}.  In
Figs.~\ref{fig:wang:init:resp}--\ref{fig:wang:ttp}, the columns of
nodes connected by double arrows $\Rightarrow$ are strands.

\begin{Assum}
  Strands and nodes are disjoint from $\Algebra$. 
\end{Assum}
A message $t_0$ \emph{originates} at a node $n_1$ if (1) $n_1$ is a
transmission node; (2) $t_0\subterm\term(n_1)$; and (3) whenever
$n_0\Rightarrow^{+}n_1$, $t_0\not\subterm\term(n_0)$.

Thus, $t_0$ originates when it was transmitted without having been
either received, transmitted, or synchronized previously on the same
strand.  Values assumed to originate only on one node in an
execution---\emph{uniquely originating} values---formalize the idea of
freshly chosen, unguessable values.  Values assumed to originate
nowhere may be used to encrypt or decrypt, but are never sent as
message ingredients.  They are called \emph{non-originating} values.
For a non-originating value $K$, $K\not\subterm t$ for any transmitted
message $t$.  However, $K\ll\enc{t_0}{K}\subterm t$ possibly, which is
why we distinguish $\subterm$ from $\ll$.
See~\cite{GuttmanThayer02,DoghmiGuttmanThayer07} for more details.

\paragraph{Protocols.} A \emph{protocol} $\Pi$ is a finite set of
strands, representing the roles of the protocol.  Their instances
result by replacing $A,B,K,M$, etc., by any names, symmetric key,
text, etc.  Each protocol also contains the \emph{listener} role
$\Hear{}{y}$ with a single reception node in which $y$ is received.
The instances of $\Hear{}{y}$ are used to document that values are
available without cryptographic protection.

Indeterminates represent messages received from protocol peers, or
passed down as parameters from higher-level protocols.  Thus, we
require:
\begin{description}
\item[If] $n_1$ is a node on $\rho\in\Pi$, with an indeterminate
  $x\ll\term(n_1)$,
\item[then] $\exists n_0$, $n_0\Rightarrow^{*}n_1$, where $n_0$ is a
  reception node and $x\subterm\term(n_0)$.
\end{description}
So, an indeterminate is received as an ingredient before appearing in
any other way.  
We say that a strand $s$ is \emph{in} $\bnd$ if $s$ has at least one
node in $\bnd$.  
\begin{prop} 
\label{lemma:order}
Let $\bnd$ be a bundle.  $\preceq_{\bnd}$ is a well-founded partial
order.  Every non-empty set of nodes of $\bnd$ has
$\preceq_{\bnd}$-minimal members.  If $a\subterm\term(n)$ for any
$n\in\bnd$, then $a$ originates at some $m\preceq_{\bnd}n$.  
\end{prop}


\end{document}